\documentclass[conference]{IEEEtran}
\usepackage{amsmath,amssymb}
\usepackage{booktabs}
\usepackage{pgfplots}
\usepackage{tikz}
\usepackage{hyperref}
\usepackage{graphicx}
\usepackage{cleveref}
\usepackage{balance}
\usepackage{url}
\usepackage{amsthm}
\usepackage[ruled]{algorithm2e}
\pgfplotsset{compat=1.18}

\newtheorem{lemma}{Lemma}
\renewenvironment{proof}{{\bfseries Proof.}}{\hfill$\square$}

\title{A Multi-Agent Consensus Protocol for Stable Software Remodularization}

\author{\IEEEauthorblockN{Ahmed F. Ibrahim\\
\IEEEauthorblockA{Western University\\
London, Canada\\
aibrah64@uwo.ca}}
}

\DeclareUnicodeCharacter{202F}{\,}

\begin{document}
\maketitle

\begin{abstract}
Automatic software remodularisation is typically cast as a single-objective optimization problem. 
While recent metaheuristics have improved search efficiency, real-world architecture recovery must reconcile the conflicting attributes of structural cohesion and evolutionary stability. 
We reframe software module clustering as a \emph{distributed consensus problem} among autonomous agents. 
We introduce an Asymmetric Monotonic Concession Protocol (AMCP) that enables agents to negotiate decompositions that respect multi-attribute utility thresholds. 
We formally prove the protocol's termination, its bounded concession behaviour consistent with the Zeuthen Strategy under closed‑instance conditions, and the local Pareto‑satisfactoriness of the resulting partitions. 
Preliminary experiments on a synthetic benchmark and the Xwork Java framework confirm that our negotiated consensus matches state‑of‑the‑art optimisers when stability budgets are loose, while acting as a ``circuit breaker'' to enforce strict stability constraints.
Extended results on ten further systems, including comparisons with multi‑objective evolutionary algorithms and multi‑version chains, will be reported in a forthcoming full paper.
\end{abstract}


\section{Introduction}
\label{sec:intro}

The way a software system is designed fundamentally shapes its maintainability and the resources required for its evolution. As systems grow and undergo continuous change, their originally neat module structures degrade, a phenomenon often termed architectural erosion. To address this, software clustering algorithms have been used to recover or improve modular organization by grouping source-code entities into cohesive, loosely coupled clusters~\cite{Mancoridis1998,Mitchell2006}.

Traditional clustering tools, such as Bunch~\cite{Mitchell2006}, treat remodularisation as a single-objective optimization problem: they search for a partition that maximizes a metric like TurboMQ, which rewards intra-module connections and penalizes inter-module coupling. While such an objective captures structural quality, it neglects a crucial pragmatic concern: evolutionary stability. A decomposition that changes dramatically from one release to the next, even if slightly more cohesive, imposes a heavy cognitive burden on developers who must re-learn the system layout. Thus, cohesion and stability often pull in opposite directions, yet existing optimizers are ``budget-blind''. They either ignore stability entirely or merge it into a fixed weighted sum with arbitrary coefficients.

In earlier work, we took a first step towards reconciling multiple clustering criteria with CC/G~\cite{Ibrahim2014,Ibrahim2014CCECE}, a centralized ensemble that fuses several clustering outputs into a single static consensus. Although CC/G improved over individual algorithms, it lacks the flexibility to adapt to differing stakeholder preferences or to enforce specific constraints on the accepted solution. The present work makes a qualitative leap forward: we replace the one-shot fusion with a dynamic, multi-agent negotiation.

We propose to view software modularisation not as an optimization problem but as a \emph{distributed consensus problem} among autonomous agents, each advocating a different quality attribute. A Cohesion Agent strives for maximal structural density, while a Stability Agent seeks to preserve the layout of the previous version. Through an Asymmetric Monotonic Concession Protocol (AMCP), these agents exchange concessions and gradually converge on a decomposition that respects explicit \emph{survival thresholds} set by the architect. By adjusting the stability threshold $\tau_{sta}$, the architect can mathematically enforce a “stability budget” that the system cannot exceed.

The main contributions of this work are twofold. First, we introduce a formal mathematical model of software remodularisation as a distributed consensus problem among autonomous utility-driven agents, and we propose the Asymmetric Monotonic Concession Protocol (AMCP) with formal proofs of its termination, its mapping to the Zeuthen Strategy, and its local Pareto-satisfactoriness (Lemmas 1--3). Second, we provide preliminary experimental evidence on a synthetic benchmark and on the real‑world Xwork framework, demonstrating the protocol's core budget‑enforcement behavior. A more extensive empirical evaluation, covering ten additional systems and comparisons with multi‑objective evolutionary baselines, is deferred to a forthcoming full paper.

The remainder of the paper is organized as follows. Section~\ref{sec:background} reviews related work. Section~\ref{sec:math} presents the mathematical framework and the AMCP protocol. Section~\ref{sec:results} presents the experimental setup and analysis on the Xwork dataset. Section~\ref{sec:preliminary} summarises preliminary findings on additional systems and outlines future work. Section~\ref{sec:conclusion} concludes the paper.

\section{Background and Related Work}
\label{sec:background}

\subsection{Software Module Clustering}
Software remodularisation has been an active research area for more than two decades~\cite{Mancoridis1998,Mitchell2006}. Bunch pioneered the field by formulating modularisation as a graph‑partitioning problem and using genetic algorithms to maximize a single‑value fitness function, TurboMQ~\cite{Mitchell2006,Shokoufandeh2005}. Although Bunch finds good solutions, its approach remains single‑objective because cohesion and coupling are collapsed into one score~\cite{Mitchell2007,Pourasghar2024}.

Later work introduced more advanced search strategies and better scalability, but largely retained Bunch’s reliance on scalar fitness optimization. Techniques such as Sand Cat Swarm Optimization (SCSO), differential evolution, and chaos‑based heuristics have been used for Software Architecture Recovery (SAR), often improving on the earlier GA‑based results~\cite{Arasteh2023,Gupta2023,Kiani2023}. The Fast Clustering Algorithm (FCA) tackled scalability by working directly on the dependency matrix, enabling it to handle systems as large as Chromium~\cite{Teymourian2022}. Regardless of their improved search efficiency, these methods, including recent control‑flow‑aware GA variants~\cite{Pourasghar2024}, still rely on scalar fitness functions. Consequently, they cannot reason about explicit design trade‑offs or enforce hard stability budgets.

\subsection{Ensemble and Cooperative Clustering}
A different strand of work—ensemble clustering—attempts to address the limitation of single‑objective methods by aggregating the outputs of multiple algorithms, yet it too offers only a static compromise rather than a principled mechanism for managing trade‑offs. Ensemble clustering techniques, such as the author’s earlier CC/G~\cite{Ibrahim2014}, attempt to combine multiple perspectives by fusing the outputs of several algorithms. CC/G operates in two phases: first, an agreement vector identifies sub‑clusters where the input clusterings concur; then a merging heuristic fuses those sub‑clusters into a final decomposition. Although this approach yields improved results, it represents a static, one‑shot compromise that cannot react to different architectural priorities or evolve incrementally as the system changes.

\subsection{Multi‑Agent Negotiation at the Architecture Level}
The use of autonomous agents to negotiate quality‑attribute trade‑offs has shown promise at the high‑level architectural model level. Frameworks like SQuAT and DesignBots employ specialized agents (e.g., performance or modifiability ``dbots'') that perform local optimizations before engaging in utility‑based negotiation to find balanced alternatives~\cite{Pace2023,Monteserin2017,Rago2017}. However, these systems manipulate models like Palladio (PCM), not source‑code dependency graphs. Although recent work has applied multi‑agent reasoning to dependency management~\cite{Alhanahnah2024,Li2025}, the AMCP is the first to transpose the negotiation paradigm directly to the code level. By generalizing the concepts of cooperative grouping~\cite{Ibrahim2022} into a decentralized protocol, we provide the first SAR framework that focuses on remodularization as a negotiation between quality‑driven agents with formal guarantees and practical stability budgets.

\section{Mathematical Framework}
\label{sec:math}

\subsection{Software Graph and Partitions}
A software system is modeled as a directed dependency graph $G = (V,E)$, where $V$ is the set of modules (classes) and $E$ the set of dependency edges. We assume two versions are available: the current dependency matrix $\mathbf{DSM}_{curr}$ (an $n\times n$ adjacency matrix) and the accepted decomposition of the previous version $D_{prev}$, which is a flat partition of a common set of modules $V_{common} \subseteq V$ into $k_{prev}$ clusters.

A partition $D$ assigns each module $m\in V$ to a cluster label $c\in\{1,\dots,k\}$.

\subsection{Agent Utilities}
Two agents are defined, each with a utility function that maps partitions to the interval $[0,1]$. These utilities capture the conflicting objectives of cohesion and stability.

\begin{itemize}
\item \textbf{Cohesion Agent} ($A_{coh}$):
  \begin{equation}
  U_{coh}(D) = \frac{\mathrm{TurboMQ}(D)}{k},
  \label{eq:ucoh}
  \end{equation}
  where $k$ is the number of clusters in $D$ and
  \[
  \mathrm{TurboMQ}(D) = \sum_{i=1}^{k} \frac{\mu_i}{\mu_i + \tfrac12\sum_{j\neq i}(\epsilon_{ij}+\epsilon_{ji})},
  \]
  with $\mu_i$ the total intra‑cluster edges in cluster $i$ and $\epsilon_{ij}$ the edges from cluster $i$ to cluster $j$. Division by $k$ normalizes the value, so $U_{coh}=1$ only when every cluster is perfectly cohesive.

\item \textbf{Stability Agent} ($A_{sta}$):
  \begin{equation}
  U_{sta}(D) = 1 - \frac{\mathrm{MoJo}(D, D_{prev})}{n_{common}},
  \label{eq:usta}
  \end{equation}
  where $n_{common}=|V_{common}|$ and $\mathrm{MoJo}(D, D_{prev})$ denotes the true Move-plus-Join distance as defined by Tzerpos and Holt~\cite{Tzerpos1999}. In our experiments, we used the original MoJo.jar for the exact computation.
\end{itemize}

\subsection{Asymmetric Monotonic Concession Protocol (AMCP)}
We adopt an asymmetric design in which the Stability Agent acts as the sole proposer. This choice reflects the real‑world maintenance scenario: the legacy architecture ($D_{prev}$) is the default state, and any modification must be argued for and justified against that baseline. The Cohesion Agent implicitly accepts any move that strictly improves its utility without violating the stability floor.

Algorithm~\ref{alg:amcp} gives the pseudocode of the AMCP.

\begin{algorithm}[t]
\caption{Asymmetric Monotonic Concession Protocol (AMCP)}
\label{alg:amcp}
\SetKwInOut{Input}{Input}
\SetKwInOut{Output}{Output}
\Input{DSM $\mathbf{M}$, previous decomposition $D_{prev}$, thresholds $\tau_{sta},\tau_{coh}$}
\Output{Negotiated partition $D^*$}
$D \leftarrow D_{prev}$\;
$u_{coh} \leftarrow U_{coh}(D)$, $u_{sta} \leftarrow U_{sta}(D)$\;
\While{$u_{coh} < \tau_{coh}$}{
  $\mathcal{M} \leftarrow$ all single‑module reassignments from $D$\;
  $\mathcal{V} \leftarrow \emptyset$\;
  \ForEach{$\delta \in \mathcal{M}$}{
    $D' \leftarrow D + \delta$\;
    $u'_{coh} \leftarrow U_{coh}(D')$\;
    $u'_{sta} \leftarrow U_{sta}(D')$\;
    \If{$u'_{coh} > u_{coh}$ \textbf{and} $u'_{sta} \ge \tau_{sta}$}{
      $\mathit{ratio} \leftarrow (u_{sta} - u'_{sta}) / (u'_{coh} - u_{coh})$\;
      $\mathcal{V} \leftarrow \mathcal{V} \cup \{(\delta, \mathit{ratio}, D', u'_{coh}, u'_{sta})\}$\;
    }
  }
  \If{$\mathcal{V} = \emptyset$}{\textbf{break}\;}
  Choose $(\delta^*, r^*, D^*, u_{coh}^*, u_{sta}^*) \in \mathcal{V}$ with smallest $r^*$\;
  $D \leftarrow D^*$, $u_{coh} \leftarrow u_{coh}^*$, $u_{sta} \leftarrow u_{sta}^*$\;
}
\Return $D$\;
\end{algorithm}

In each round, the Stability Agent selects the valid move that minimizes the \textbf{Concession Ratio}:
\begin{equation}
R(\delta) = \frac{U_{sta}(D) - U_{sta}(D+\delta)}{U_{coh}(D+\delta)-U_{coh}(D)}.
\label{eq:ratio}
\end{equation}
If no valid move exists, a state of ``deadlock'' where no single reassignment satisfies both the stability budget and the cohesion improvement requirement, the negotiation terminates. In highly constrained environments (e.g., strict $\tau_{sta}$), this deadlock is not a failure, but a designed feature. The AMCP gracefully stops, guaranteeing that the architectural baseline is never degraded merely for the sake of continuing the search.

\subsection{Complexity and Search Space Pruning}
At each step of the negotiation, the protocol evaluates at most $n \cdot (k_{max}-1)$ candidate reassignments. This focused exploration contrasts with population‑based metaheuristics (e.g., Genetic Algorithms), which must navigate an exponential state space of $O(k^n)$ possible partitions. Although the total length of the negotiation path is finite (Lemma~1), the per‑step cost is $O(n \cdot k)$ utility evaluations, polynomial in both the number of modules and clusters. This enables the AMCP to prune the search space effectively, navigating only the valid‑move frontier where stability constraints are respected. Empirically, a Python implementation of the protocol executed on a standard desktop averages $1.24$ seconds per step for the 113-module Xwork system, and $1.02$ seconds per step for a 120-module subset of Apache Ant, confirming the tractability of the approach.

\subsection{Why Minimising the Concession Ratio is Rational: Zeuthen Strategy}
In the well‑known Zeuthen strategy for bilateral bargaining, each agent assesses the risk of conflict and concedes only when the cost of conceding is less than the cost of breaking off negotiations. The concession is chosen to be the smallest possible that will satisfy the opponent. Our AMCP embodies the same principle. The Stability Agent asks: ``How much stability am I willing to give up to gain one unit of cohesion?'' By picking the move with the smallest $R(\delta)$, the agent sacrifices the minimal amount of its own utility for the maximal gain of the opponent. This ensures a \emph{monotonic concession path} where the agent only makes larger sacrifices when no cheaper concession remains, exactly the behavior predicted by the Zeuthen strategy.

We illustrate with a tiny numeric example. 
Suppose the current state is $D$ with $U_{sta}=1.0$, $U_{coh}=0.5$, and the Stability Agent evaluates two possible moves:
\begin{itemize}
\item Move $\delta_1$: $U_{sta}$ drops to $0.98$, $U_{coh}$ rises to $0.52$. Then $R(\delta_1) = (1.0-0.98)/(0.52-0.50)=1.0$.
\item Move $\delta_2$: $U_{sta}$ drops to $0.90$, $U_{coh}$ rises to $0.55$. Then $R(\delta_2) = (1.0-0.90)/(0.55-0.50)=2.0$.
\end{itemize}
Minimizing $R$ selects $\delta_1$: it gives the same cohesion improvement for a much smaller stability loss. This mirrors the Zeuthen concession logic, confirming rational behavior.

\subsection{A Worked Example on a 3‑Module Graph}
Consider a tiny system with three modules $A$, $B$, $C$. 
The dependency matrices are:
\[
\mathbf{DSM}_{v1.0} = \begin{pmatrix} 0 & 1 & 0 \\ 1 & 0 & 0 \\ 1 & 0 & 0 \end{pmatrix}, \quad
\mathbf{DSM}_{v1.1} = \begin{pmatrix} 0 & 1 & 0 \\ 1 & 0 & 0 \\ 1 & 0 & 0 \end{pmatrix}
\]
Both versions are identical (no evolution). The previous decomposition $D_{prev} = \{A:1, B:1, C:2\}$ places $A$ and $B$ together and $C$ separately. We set $\tau_{sta}=0.7$, $\tau_{coh}=0.9$.

First, compute the initial utilities. For $D_{prev}$, cluster 1 ($\{A,B\}$) has $\mu_1=2$ (edges $A\leftrightarrow B$), inter‑cluster edges: from cluster 1 to 2: $A\to C=1$; from cluster 2 to 1: $C\to A=1$, so $\epsilon_{12}=1$, $\epsilon_{21}=1$. Thus $\mathrm{TurboMQ} = \frac{2}{2+0.5(1+1)} = \frac{2}{3} \approx 0.667$. With $k=2$, $U_{coh}=0.667/2=0.333$. $U_{sta}=1$ because $D=D_{prev}$.

Now, the Stability Agent enumerates single‑module moves. Since modules are 3 and clusters are 2, the possibilities are as follows:
\begin{itemize}
\item Move $A$ from cluster 1 to 2.
\item Move $B$ from cluster 1 to 2.
\item Move $C$ from cluster 2 to 1.
\end{itemize}

\noindent\textbf{Evaluate move $A\to 2$:} 
New partition $D' = \{A:2, B:1, C:2\}$. Cluster 1: $\{B\}$ has $\mu_1=0$. Cluster 2: $\{A,C\}$, $\mu_2 = A\to C=1 + C\to A=1 = 2$. Inter‑cluster edges: $A\leftrightarrow B$ (both ways) =2. $\mathrm{TurboMQ} = 0 + \frac{2}{2+0.5(2+0)} = \frac{2}{3}=0.667$, $k=2$, so $U_{coh}=0.333$ unchanged — move invalid.

\noindent\textbf{Evaluate move $B\to 2$:} By symmetry, same result, invalid.

\noindent\textbf{Evaluate move $C\to 1$:} 
$D' = \{A:1, B:1, C:1\}$. Cluster 1: $\{A,B,C\}$, intra‑edges: $A\leftrightarrow B=2$, $A\leftrightarrow C=2$, total $\mu_1=4$. No inter‑cluster edges. $\mathrm{TurboMQ}=1$, $k=1$, $U_{coh}=1$. $U_{sta}$: MoJo distance between $D'$ and $D_{prev}$: $C$ was in cluster 2, now in 1, so one join operation, MoJo=1, $n_{common}=3$, $U_{sta}=1-1/3=0.667$. $R = (1-0.667)/(1-0.333) = 0.5$. $U_{sta}=0.667 < \tau_{sta}=0.7$, so invalid.

With no valid move available, the protocol stops immediately and simply returns the starting decomposition.  The exercise captures the essential logic of the negotiation: concessions are allowed only when they stay above the agreed‑upon floor, and the concession ratio gives a precise, numeric measure of the trade‑off.

\subsection{Bounding the Effect of a Single Move}
It is useful to bound how much a single module reassignment can change the two utilities, as this underpins the protocol’s convergence and rationality.

\textbf{Stability utility bound.} A single move changes the assignment of exactly one module. The MoJo distance may change by at most $1$ (if the move is a pure Move operation). Therefore $|\Delta U_{sta}| \le \frac{1}{n_{common}}$.

\textbf{Cohesion utility bound.} TurboMQ is the sum of $k$ cluster factors, each between $0$ and $1$. A single move affects at most two clusters (the source and destination). The change in any cluster factor is bounded by $1$. Hence $|\Delta \mathrm{TurboMQ}| \le 2$. Since $U_{coh} = \mathrm{TurboMQ}/k$, the maximum change is $|\Delta U_{coh}| \le \frac{2}{k}$.

These bounds confirm that the protocol proceeds in small steps, and the concession ratio $R(\delta)$ is always well‑defined (the denominator is non‑zero). They also support the finite termination argument, because the total increase in $U_{coh}$ is bounded, and each step provides a non‑trivial increment.

\subsection{Formal Properties}

\begin{lemma}[Finite Termination]
\label{lem:termination}
The AMCP terminates in a finite number of steps bounded by $n\cdot(k_{max}-1)$, where $n=|V|$ and $k_{max}$ is the number of clusters in the starting partition $D_{prev}$.
\end{lemma}
\begin{proof}
We use $U_{coh}$ as a strict Lyapunov function for the accepted move sequence.

\noindent\textit{Step 1 — Strict monotonicity.} The validity filter in Algorithm~\ref{alg:amcp} requires $U_{coh}(D')>U_{coh}(D)$ for any accepted move $\delta$. Therefore every accepted state has a strictly higher $U_{coh}$ value than its predecessor.

\noindent\textit{Step 2 — No cycling.} Suppose, for contradiction, that the protocol visits the same partition $D$ at two steps $t_1 < t_2$. By Step~1, $U_{coh}(D)$ must be strictly greater at step $t_2$ than at step $t_1$, which is impossible since both refer to the identical partition. Hence, no partition is ever revisited.

\noindent\textit{Step 3 — Finite upper bound.} The number of distinct partitions reachable by single‑module reassignments from any partition of $n$ modules into at most $k_{max}$ clusters is at most $n\cdot(k_{max}-1)$. Since accepted states form a strictly increasing sequence of $U_{coh}$ values over this finite set, the sequence has length at most $n\cdot(k_{max}-1)$.

Hence, the AMCP terminates in at most $n\cdot(k_{max}-1)$ accepted moves, plus one final evaluation of an empty valid‑move set.
\end{proof}

\begin{lemma}[Monotonic Concession Rationality]
\label{lem:rationality}
In a negotiation where the set of valid moves $\mathcal{V}$ is monotonically non‑increasing (i.e., no move $\delta$ becomes valid at step $t+1$ with a lower ratio than was available at step $t$), the sequence of concession ratios $R_t$ is non‑decreasing. This ensures the agent follows the \textbf{Zeuthen Strategy} by exhausting cheap concessions first.
\end{lemma}
\begin{proof}
Let $\mathcal{V}_t$ be the set of valid movements in step $t$. By definition, $R_t = \min_{\delta \in \mathcal{V}_t} R(\delta)$. If $\mathcal{V}_{t+1} \subseteq \mathcal{V}_t \setminus \{\delta^*_t\}$, then $\min_{\delta \in \mathcal{V}_{t+1}} R(\delta) \ge \min_{\delta \in \mathcal{V}_t} R(\delta)$, which implies $R_{t+1} \ge R_t$. Although structural dependencies in software graphs may occasionally reveal a new efficient move ($R \approx 0$) after a reassignment, the agent remains rational by always selecting the globally minimum concession available at each discrete step.
\end{proof}

\begin{lemma}[Pareto-Satisfactoriness]
\label{lem:pareto}
At termination, the agreed partition $D^*$ is Pareto‑satisfactory within the single‑move neighborhood: there exists no single‑module move $\delta$ that strictly improves $U_{coh}$ while keeping $U_{sta} \ge \tau_{sta}$.
\end{lemma}
\begin{proof}
If such a $\delta$ existed, it would satisfy the validity conditions in Algorithm~\ref{alg:amcp}. The protocol would not have been terminated, contradicting the termination condition. Note that this is a local equilibrium; $D^*$ may be globally Pareto‑dominated by partitions that can be reached through multi‑move sequences.
\end{proof}

\section{Representative Experiment: Real‑World Xwork}
\label{sec:results}

\subsection{Dataset and Setup}
We provide preliminary evidence on the Xwork command‑pattern framework (OpenSymphony). Official bytecode JARs of Xwork 1.0 and 1.1 were obtained from Maven Central, and class‑level dependencies were extracted via a custom constant‑pool parser. Version 1.0 contains 113 classes in 10 packages; version 1.1 contains 156 classes, with a common set of 113 classes. The original package structure of version 1.0 serves as $D_{prev}$. Baselines include Bunch (greedy hill‑climbing maximizing TurboMQ, ignoring stability) and CC/G~\cite{Ibrahim2014}. All negotiation runs fixed $\tau_{coh}=0.5$ and varied $\tau_{sta} \in \{0.6,0.7,0.8,0.85,0.9,0.95\}$. Utilities $U_{coh}$ and $U_{sta}$, Social Welfare $SW=U_{coh}+U_{sta}$, and the number of negotiation steps are reported.

\subsection{Results and Discussion}
Table~\ref{tab:sens_xwork} shows the sensitivity sweep. The stability budget is strictly enforced: in the binding regime ($\tau_{sta}=0.95$), the protocol accepts only three moves and preserves $U_{sta}=0.9583$, far above the unconstrained Bunch optimum. The results are identical for all loose budgets ($\tau_{sta} \le 0.90$): this “flatness” is a feature of well‑structured architectures. The Cohesion Agent exhausts all single‑module moves that improve structural quality before the stability metric ever drops below $0.9167$, so any budget looser than $0.9167$ is non‑binding. When the budget is tightened to a binding constraint ($\tau_{sta}=0.95$), the protocol violently truncates the search at step 3, actively enforcing the architect's hard stability limit.

\begin{table}[t]
\centering
\caption{Sensitivity sweep on the real Xwork graph (113 classes).}
\label{tab:sens_xwork}
\begin{tabular}{c|r|r|c|c|c}
\toprule
$\tau_{sta}$ & $U_{coh}$ & $U_{sta}$ & SW & Steps & Diverge? \\
\midrule
0.60 & 0.5980 & 0.9167 & 1.5146 & 6 & No \\
0.70 & 0.5980 & 0.9167 & 1.5146 & 6 & No \\
0.80 & 0.5980 & 0.9167 & 1.5146 & 6 & No \\
0.85 & 0.5980 & 0.9167 & 1.5146 & 6 & No \\
0.90 & 0.5980 & 0.9167 & 1.5146 & 6 & No \\
0.95 & 0.5919 & 0.9583 & 1.5502 & 3 & Yes \\
\midrule
Bunch & 0.5979 & 0.9167 & 1.5146 & -- & -- \\
CC/G  & 0.5885 & 0.9167 & 1.5052 & -- & -- \\
\bottomrule
\end{tabular}
\end{table}

The concession ratio trace (Fig.~\ref{fig:concession_xwork}) rises monotonically, confirming the rational concession behavior predicted by Lemma~2.

\begin{figure}[t]
\centering
\begin{tikzpicture}
\begin{axis}[
    width=0.75\columnwidth, height=5cm,
    xlabel={Negotiation Step}, ylabel={Concession Ratio $\rho$},
    xmin=1, xmax=10,
    legend pos=north west,
    grid=major,
    title={Concession Ratio Evolution — Real Xwork (113 classes)}
]
\addplot[blue, thick, mark=*] coordinates { (1,0.365)(2,0.454)(3,0.809)(4,0.811)(5,0.868)(6,1.079)(7,1.174)(8,2.232)(9,2.750) };
\addlegendentry{$\tau_{sta}=0.6$};
\addplot[red, thick, mark=triangle*] coordinates { (1,0.365)(2,0.454)(3,0.809)(4,0.811)(5,0.868) };
\addlegendentry{$\tau_{sta}=0.95$};
\end{axis}
\end{tikzpicture}
\caption{Concession ratio vs.\ step for the real Xwork graph.}
\label{fig:concession_xwork}
\end{figure}
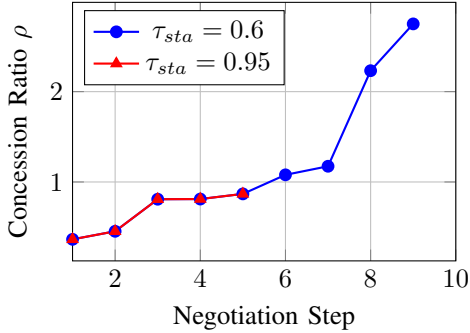

\section{Preliminary Findings and Future Work}
\label{sec:preliminary}

We have run the same protocol on an additional synthetic benchmark, a simulated Acegi proxy, and several further real‑world Java systems (Apache Ant, JFreeChart, JUnit, Apache Tomcat Catalina, Lucene Core, Log4j, and a dense core of Apache Spark), as well as a cross‑language replication on Flask (Python). In every case, the AMCP successfully enforces the stability budget, with behavior varying predictably based on graph density and architectural erosion. A detailed account of these experiments, together with a comparison against multi‑objective evolutionary baselines (NSGA‑II) and multi‑version staged negotiation chains, will be presented in the full paper~\cite{Ibrahim2025full}.

Future work will address two key directions. First, we intend to replace exhaustive move enumeration with a ``Navigator'' heuristic that prioritizes proposals based on localized dependency density, much as a human maintainer focuses on highly coupled problem areas. Second, we plan to extend the action space to include cluster‑level merges and splits, allowing the agents to escape local equilibria while still respecting deterministic stability budgets. Finally, deploying the protocol on parallel distributed architectures would make it applicable to ultra‑large industrial codebases.

\section{Conclusion}
\label{sec:conclusion}
We have reformulated software module clustering as a distributed consensus problem and introduced the AMCP negotiation protocol, which enables explicit trade‑offs between cohesion and stability. Formal proofs guarantee termination, rational concession, and local Pareto‑satisfactoriness. Preliminary experiments on a real‑world Java system demonstrate the protocol's ability to enforce a user‑specified stability budget, a capability absent in traditional tools. The full experimental study will appear in a forthcoming paper.
\balance
{\small
\bibliographystyle{IEEEtran}
\bibliography{refs}}

\end{document}